\documentclass[runningheads]{llncs}
\usepackage[margin=1in, bottom=1in, top=1in]{geometry} 
\usepackage{amsmath, amssymb, amstext}
\usepackage{fancyhdr}
\usepackage{mathtools}
\usepackage{algorithm}
\usepackage{algpseudocode}
\usepackage{theorem}
\usepackage{tikz}
\usepackage{tkz-berge}
\usepackage{authblk}
\usepackage[colorlinks]{hyperref}
\usepackage{enumerate}  
\usepackage{bbm}
\usepackage{mathbbol}
\usepackage{comment}
\usepackage{wrapfig}

 \newcommand{\F}{\mathbb{F}}
 \newcommand{\R}{\mathbb{R}}
 \newcommand{\Z}{\mathbb{Z}}

 \newcommand{\cB}{\mathcal{B}}

 \newcommand{\cL}{\mathcal{L}}
 
 \newcommand{\cP}{\mathcal{P}}

 \newcommand{\cX}{\mathcal{X}}

\newcommand{\fix}{\mathrm{Fix}}

\DeclareMathOperator{\opspan}{span}

\DeclareMathOperator{\treewidthOp}{tw}
\newcommand{\tw}{\treewidthOp}

\DeclarePairedDelimiter\abs{\lvert}{\rvert}%

\usepackage{graphicx}

\begin{document}%
\title{A General Framework for Computing the Nucleolus Via Dynamic Programming\thanks{We acknowledge the support of the Natural Sciences and Engineering Research Council of Canada (NSERC). Cette recherche a \'et\'e financ\'ee par le Conseil de recherches en sciences naturelles et en g\'enie du Canada (CRSNG).}}
\titlerunning{Nucleolus Dynamic Programming}
%
\author{Jochen K\"{o}nemann\inst{1} \and
Justin Toth\inst{1}}
\authorrunning{J. K\"{o}nemann et al.}
%
\institute{University of Waterloo, Waterloo ON N2L 3G1,  Canada 
\email{\{jochen, wjtoth\}@uwaterloo.ca}}
\maketitle              

\begin{abstract}
  This paper defines a general class of cooperative games for which the nucleolus is efficiently computable. This class includes new members for which the complexity of computing their nucleolus was not previously known. We show that when the minimum excess coalition problem of a cooperative game can be formulated as a hypergraph dynamic program its nucleolus is efficiently computable. This gives a general technique for designing efficient algorithms for computing the nucleolus of a cooperative game. This technique is inspired by a recent result of Pashkovich~\cite{Pashkovich2018} on weighted voting games. However our technique significantly extends beyond the capabilities of previous work. We demonstrate this by applying it to give an algorithm for computing the nucleolus of $b$-matching games in polynomial time on graphs of bounded treewidth.
\keywords{Combinatorial Optimization  \and  Algorithmic Game Theory \and Dynamic Programming}
\end{abstract}

\section{Introduction and Related Work}\label{sec:intro}
\paragraph{}
Cooperative game theory studies situations in which individual agents form coalitions to work together towards a common goal. It studies questions regarding what sort of coalitions will form and how they will share the surplus generated by their collective efforts. A \emph{cooperative game} is defined by an ordered pair $([n],\nu)$ where $[n]$ is a finite set of players (labelled $1,\dots, n$), and $\nu$ is a function from subsets of $[n]$ to $\R$ indicating the value earned by each particular coalition. 

This paper studies the computational complexity of one of the most classical, deep, and widely applicable solution concepts for surplus division in cooperative games, the \emph{nucleolus}. In particular we study the relationship between the nucleolus, finding the minimum excess of a coalition, congruency-constrained optimization, and dynamic programming. Our first result unifies these areas and provides a general method for computing the nucleolus.

\begin{theorem}\label{th:dp-nucleolus}
    For any cooperative game $(n,\nu)$, if the minimum excess coalition problem on $(n,\nu)$ can be solved in time $T$ via an integral dynamic program then the nucleolus of $(n,\nu)$ can be computed in time polynomial in $T$.
\end{theorem}

Pashkovich~\cite{Pashkovich2018} showed how to reduce the problem of computing the nucleolus for weighted voting games to a congruency-constrained optimization problem. Pashkovich then shows how to solve this congruency-constrained optimization problem for this specific class of games via a dynamic program. In Section~\ref{sec:linear-subspace} we abstract his reduction to the setting of computing the nucleolus of general combinatorial optimization games.

Our main technical achievement is showing that adding congruency constraints to dynamic programs modelled by a directed acyclic hypergraph model inspired by the work of Campbell, Martin, and Rardin~\cite{martin1990polyhedral} adds only a polynomial factor to the computational complexity. This is the content of Theorem~\ref{th:congruency}, which is instrumental in demonstrating Theorem~\ref{th:dp-nucleolus}. Our formal model of dynamic programming, where solutions correspond to directed hyperpaths in a directed acyclic hypergraph, is described in Section~\ref{sec:dynamic-programming}. Proving Theorem~\ref{th:congruency} requires significant new techniques beyond~\cite{Pashkovich2018}. The series of lemmas in Section~\ref{sec:congruency-constrained-dp} take the reader through these techniques for manipulating directed acyclic hypergraph dynamic programs.

We show how Theorem~\ref{th:dp-nucleolus} not only generalizes previous work on computing the nucleolus, but significantly extends our capabilities to new classes of combinatorial optimization games that were not possible with just the ideas in~\cite{Pashkovich2018}. As we explain in Section~\ref{sec:comb-opt-games}, matching games are central to the study of combinatorial optimization games. The problem of computing the nucleolus of weighted matching games was a long-standing open problem~\cite{Faigle1998}~\cite{Kern2003} resolved only recently~\cite{Konemann2020}, nearly twenty years after it was first posed. The frontier for the field has now moved to $b$-matching games, for which computing the nucleolus is believed to be NP-hard in general due to the result in~\cite{Biro2017} which shows computing leastcore allocations to be NP-hard even in the unweighted, bipartite case with $b\equiv 3$. In Section~\ref{sec:applications} we give a result which significantly narrows the gap between what is known to be tractable and what is known to be intractable in that area.

\begin{theorem}\label{th:b-matching-nucleolus}
    For any cooperative $b$-matching game on a graph whose treewidth is bounded by a constant, the nucleolus can be computed in polynomial time.
\end{theorem}

To achieve this result we give a dynamic program for computing the minimum excess coalition of a $b$-matching game in Lemma~\ref{lemma:min-excess-b-matching-specific-dp} then apply Theorem~\ref{th:dp-nucleolus}. This dynamic program necessarily requires the use of dynamic programming on hypergraphs instead of just simple graphs, motivating the increased complexity of our model over previous work.

\subsection{The Nucleolus}

When studying the question of surplus division, it is commonly desireable that shares will be split so all players have an incentive to work together, i.e.\ that the \emph{grand coalition} forms. A vector $x \in \R^n$ is called an \emph{allocation}, and if that vector satisfies $x([n]) = \nu([n])$ (efficiency) and $x_i \geq \nu(\{i\})$, for all $i\in[n]$, (individual rationality) we call $x$ an \emph{imputation}. We denote the set of imputations of $(n,\nu)$ by $I(n,\nu)$.

For any $S \subseteq [n]$ we define $x(S)-\nu(S)$ to be the \emph{excess} of $S$ with respect to allocation $x$. The following linear program maximizes the minimum excess:
\begin{align*}\label{lp:leastcore}\tag{$P_1$}
    \max&\ \epsilon \\
    \text{s.t. }& x(S) \geq \nu(S) + \epsilon &\forall S\subseteq [n]\\
    &x \in I(n,\nu)
\end{align*} 
We call this the \emph{leastcore linear program}. For any $\epsilon$ be let $P_1(\epsilon)$ denote the set of allocations $x$ such that $(x,\epsilon)$ is a feasible solution to (\ref{lp:leastcore}). If we let $\epsilon_1$ denote the optimal value of (\ref{lp:leastcore}) then we call $P_1(\epsilon_1)$ the \emph{leastcore}\footnote{In the literature it is more standard to define the leastcore without individual rationality constraints. Since individual rationality constraints rarely cause additional computational difficutly, we include them so that the leastcore's relationship to the nucleolus is more evident.} of the cooperative game.

For an imputation $x \in I(n,\nu)$ let $\theta(x)\in \R^{2^n-2}$ be the vector obtained by sorting the list of excess values $x(S) - \nu(S)$, for each $\emptyset\neq S \subset [n]$, in non-decreasing order. \begin{definition}The {nucleolus} is the imputation which lexicographically maximizes $\theta(x)$, formally: 
the nucleolus is equal to $\arg \text{lex} \max \{\theta(x): x \in I(n,\nu)\}.$
\end{definition}
The nucleolus was first defined by Schmeidler~\cite{Schmeidler1969}. In the same paper, Schmeidler showed the nucleolus to be a unique allocation and a continuous function of $\nu$. The nucleolus is a classical object in game theory, attracting attention for its geometric beauty~\cite{Maschler1979}, and its surprising applications. The most ancient of which is the application of the nucleolus as a bankrupcy division scheme in the Babylonian Talmud~\cite{Aumann1985}. Some other notable applications of the nucleolus include but are not limited to water supply management~\cite{Akbari2015}, fair file sharing on peer-to-peer networks~\cite{Militano2013}, resource sharing in job assignment~\cite{Solymosi2005}, and airport pricing~\cite{Branzei2006}.

\subsection{Computing the Nucleolus}

Multiple approaches exist for algorithmically finding the nucleolus of a cooperative game. The most ubiquitous of which is the MPS (Maschler, Peleg, Shapley) Scheme~\cite{Maschler1979} which operates by solving a hierarchy of at most $n$ linear programs, the last of which has the nucleolus as its unique optimal solution. In Section~\ref{sec:schemes} we elaborate on the MPS Scheme and a natural relaxation thereof. An alternative method of computing the nucleolus via characterization sets was proposed independently by Granot, Granot, and Zhu~\cite{Granot1998} and Reinjerse and Potters~\cite{Reijnierse1998}.

The complexity of computing the nucleolus varies dramatically depending on how the cooperative game $(n,\nu)$ is presented as input. If the function $\nu$ is presented explicitly, by giving as input the value of $\nu(S)$ for each $S\subseteq[n]$, then the MPS Scheme can be used to compute the nucleolus in polynomial time. The issue in this case is that the size of the specification of $\nu$ is exponential in the number of players and so the computation is trivial. We are interested cooperative games where $\nu$ can be determined implicitly via some auxiliary information given as input, which we call a \emph{compact representation} of $(n,\nu)$.

One prominent example of a cooperative game with a compact representation is the class of \emph{weighted voting games}. In a weighted voting game, each player $i \in [n]$ is associated with an integer weight $w_i\in \Z$. Additionally a threshold value $T\in \Z$ is given. For each $S \subseteq [n]$ the value of $\nu(S)\in \{0,1\}$ is $1$ if and only if $w(S) \geq T$.

It is not hard to see that $(n,\nu)$ is completely determined by $(w,T)$. In this case $(w,T)$ is a compact representation of the weighted voting game $(n,\nu)$. Even though they may appear simple at first, weighted voting games can have a lot of modelling power. In fact the voting system of the European Union can be modelled by a combination of weighted voting games~\cite{bilbao2002voting}. In~\cite{Elkind2009a} Elkind, Goldberg, Goldberg, and Wooldridge show that the problem of computing the nucleolus of a weighted voting game is NP-hard, in fact even the problem of testing if there is a point in the leastcore of a weighted voting game that assigns a non-zero payoff to a given player is NP-complete. Pashkovich~\cite{Pashkovich2018} later followed up with an algorithm based on the MPS Scheme which solves $O(n)$ linear programs, each in pseudopolynomial time, and thus computes the nucleolus of a weighted voting game in pseudopolynomial time. 

Pashkovich's result crucially relies on the existence of a well-structured dynamic program for knapsack cover problems which runs in pseudopolynomial time. Theorem~\ref{th:dp-nucleolus} and Section~\ref{sec:applications} place Pahskovich's algorithm in the context of a general framework for computing the nucleolus of cooperative games where a natural associated problem has a dynamic program: the \emph{minimum excess coalition problem}. 
\begin{definition}
In the minimum excess coalition problem the given input is a compact representation of a cooperative game $(n,\nu)$ and an imputation $x$. The goal is to output a coalition $S \subseteq [n]$ which minimizes excess, i.e.\ $x(S) - \nu(S)$, with respect to $x$.
\end{definition}
\subsection{Combinatorial Optimization Games}\label{sec:comb-opt-games}

A very general class of cooperative games with compact representations comes from the so-called \emph{cooperative combinatorial optimization games}. In games of this class some overarching combinatorial structure is fixed on the players, and for each subset $S$ of players, $\nu(S)$ can be determined by solving an optimization problem on this structure. Many classes of combinatorial optimization games can be defined and the complexity of their nucleoli have been studied leading to polynomial time algorithms, such as fractional matching, cover, and clique games~\cite{chen2012computing}, simple flow games~\cite{PRB06}, assignment games\cite{SR94} and matching games~\cite{Kern2003}. Fleiner, Solymosi and Sziklai used the concept of dually essential coalitions~\cite{Solymosi2016} to compute the nucleolus of a large class of directed acyclic graph games~\cite{Sziklai2017} via the characterization set method. Other cases have led to NP-hardness proofs, such as flow games~\cite{Deng2009a}, weighted voting games~\cite{Elkind2009a}, and spanning tree games~\cite{FKK98}. 

A prominent example of combinatorial optimization games is matching games. In matching games the players are vertices of a graph $G$ and $\nu(S)$ is equal to the size of the largest matching on $G[S]$. The question of whether or not the nucleolus of weighted matching games could always be computed in polynomial time was open for a long time~\cite{Faigle1998}~\cite{Kern2003}. Solymosi and Raghavan~\cite{SR94} gave an algorithm for computing the nucleolus of matching games on bipartite graphs. Bir\'o, Kern and Paulusma gave a combinatorial algorithm for computing the nucleolus of weighted matching games with a non-empty core~\cite{Biro2012}. Recently Koenemann, Pashkovich, and Toth~\cite{Konemann2020} resolved the question by giving a compact formulation for each linear program in the MPS Scheme for weighted matching games with empty core. 

A natural generalization of matching games is to weighted $b$-matching games. In weighted $b$-matching games a vector $b \in \Z^{V(G)}$ and vector $w \in \R^{E(G)}$ are given in addition to the graph $G$. The value of $S \subseteq [n]$ is equal to the maximum $w$-weight subset of edges in $G[S]$ such that each playere $v \in V(G)$ is incident to at most $b_v$ edges. In~\cite{Biro2017} they show how to test if an allocation is in the core of $b$-matching games when $b\leq 2$, and they show that for matching games where $b\equiv 3$ deciding if an allocation is in the core is coNP-complete. This result likely means that computing the nucleolus of $b\equiv 3$-matching games is NP-hard. In~\cite{Konemann2020} they show how to separate over the leastcore of any $b\leq 2$-matching game. In~\cite{biro2019generalized} Biro, Kern, Palvolgyi, and Paulusma show that computing the nucleolus of $b$-matching games is NP-hard by showing that testing if the core of such games is empty is NP-hard. Their construction uses non-bipiartite graphs and $b$-values of size at least $3$. Hence the question of the complexity of computing the nucleolus of $b$-matching games remains open for bipartite graphs, and for $b$-matching games where $b\leq 2$. By the preceding complexity discussion, it is necessary to impose some structure on $b$-matching games to compute their nucleolus in polynomial time. In Theorem~\ref{th:b-matching-nucleolus} we impose the structure of bounded treewidth and use our general framework to give an algorithm which computes the nucleolus of weighted $b$-matching games on graphs which have bounded treewidth.

\section{The Maschler Peleg Shapley Scheme}~\label{sec:schemes}

The most prominent technique for computing the nucleolus is the MPS Scheme~\cite{Maschler1979}. The first technique for computing the nucleolus is the Kopelowitz Scheme~\cite{Kopelowitz1967}, and the MPS Scheme is a modification thereof which runs in a linear number of rounds. To define the MPS Scheme we need the notion of a fixed set for a polyhedron. For any polyhedron $Q$, we define the set $\text{Fix}(Q)$ as 
$$\text{Fix}(Q) :=\{S \subseteq [n]: \exists c \in R \text{ such that }\forall x \in Q, x(S) = c \}.$$
In the MPS Scheme a sequence of linear programs $(P_1), (P_2),\dots, (P_N)$ is computed where the $i^\text{th}$ linear program ($i\geq 2$) is of the form 
\begin{align*}\label{lp:maschler}\tag{$P_i$}
\max\ &\epsilon \\
\text{s.t. } &x(S) \geq \nu(S) + \epsilon &\forall S \not\in \text{Fix}(P_{i-1}(\epsilon_{i-1})) \\
&x \in P_{i-1}(\epsilon_{i-1}),
\end{align*}
and the first linear program is the leastcore linear program (\ref{lp:leastcore}). The method terminates when the optimal solution is unique (yielding the nucleolus), and this happens after at most $n$ rounds~\cite{paulusma2001complexity}, since the dimension of the set of characteristic vectors of sets in $\fix (P_{i}(\epsilon_i))$ increases by at least one in each iteration. 
\paragraph{}
Since the MPS Scheme ends after at most $n$ linear program solves, the run time of the method is dominated by the time it takes to solve (\ref{lp:maschler}). To use the Ellipsoid Method~\cite{Khachiyan1980,Leung1989} to implement the MPS Scheme we need be able to separate over the constraints corresponding to all coalitions in $\fix (P_{i-1}(\epsilon_{i-1}))$ in each iteration. There can be an exponential number of such constraints in general, and some structure on the underlying cooperative game would need to be observed in order to separate these constraints efficiently. This requirement can be relaxed somewhat, and still retain the linear number of iterations required to compute the nucleolus. 
\subsection{The Relaxed MPS Scheme}
We will define a sequence of linear programs $Q_1,Q_2,\dots, Q_N$ where the unique optimal solution to $Q_N$ is the nucleolus of $(n,\nu)$. With each linear program $Q_i$ there will be an associated set of vectors $V_i$ contained in the set of incidence vectors of $\fix(Q_i)$. The feasible solutions to $Q_i$ will lie in $\R^n\times \R$. In keeping with the notion we used for $P_i(\epsilon_i)$, for each linear program $Q_i$ we let $\bar{\epsilon}_i$ be the optimal value of $Q_i$ and let 
$$Q_i(\bar{\epsilon}_i) := \{x \in \R^n : (x,\bar{\epsilon}_i) \text{ is feasible for } Q_i\}.$$

We will describe the linear programs $\{Q_i\}_i$ inductively. The first linear program is again the leastcore linear program of $(n,\nu)$. That is to say $Q_1$ is equal to (\ref{lp:leastcore}). Let $V_1\subseteq \R^n$ be a singleton containing the incidence vector of one coalition in $\fix(Q_1(\bar{\epsilon}_1))$. Now given $Q_{i-1}$ and $V_{i-1}$ we describe $Q_i$ as follows 
\begin{align*}\label{lp:relaxed-maschler}\tag{$Q_i$}
    \max\ &\epsilon \\
    \text{s.t. } &x(S) \geq \nu(S) + \epsilon &\forall S: \chi(S) \not\in \opspan(V_{i-1})\\
    &x \in Q_{i-1}(\bar{\epsilon}_{i-1}).
\end{align*}
Now we choose $v \in \fix(Q_i(\bar{\epsilon}_i))\backslash \opspan(V_{i-1})$ and set $V_i := V_{i-1} \cup \{v\}$. By the optimality of $\bar{\epsilon}_i$, $v$ always exists as long as $Q_i(\bar{\epsilon}_i)$ has affine dimension at least $1$. If $Q_i(\bar{\epsilon}_i)$ has affine dimension $0$ we terminate the procedure and conclude that $Q_i(\bar{\epsilon_i})$ is a singleton containing the nucleolus.

A nice proof of correctness for this scheme is given in~\cite{Pashkovich2018}, where this scheme is used to give a pseudopolynomial time algorithm for computing the nucleolus of weighted voting games.

\begin{lemma}
When the Relaxed MPS Scheme is run on a cooperative game $(n,\nu)$ yielding a hierarchy of linear programs $Q_1,\dots, Q_N$, with optimal values $\bar{\epsilon_1},\dots, \bar{\epsilon_N}$ respectively, the set $Q_N(\bar{\epsilon_N})$ is a singleton containing the nucleolus of $(n,\nu)$. Moreover $N$ is at most $n$.
\end{lemma}

\section{The Linear Subspace Avoidance Problem}\label{sec:linear-subspace}

Motivated by the desire to design a separation oracle for the constraints of (\ref{lp:relaxed-maschler}) we initiate a general study of combinatorial optimization problems whose feasible region avoids a linear subspace. For our purposes, we say a \emph{combinatorial optimization problem} is an optimization problem of the form
\begin{equation}\tag{$P$}\label{prob:comb-opt}
\max\{f(x): x \in \cX\}
\end{equation}
where $\cX \subseteq \{0,1\}^n$ is known as the feasible region, and $f:\cX \rightarrow \R$ is the objective function. Normally (\ref{prob:comb-opt}) is presented via a compact representation. For example in the shortest path problem on a directed graph, $\cX$ is the family of paths in a directed graph $D$ and $f(x)$ is a linear function. The entire feasible set $\cX$ is uniquely determined by the underlying directed graph $D$, and $f$ is determined by weights on the arcs of $D$. When giving as input $D$ and the arc weights, the problem is completely determined without specifying every one of the exponentially many paths in $\cX$.

For compactly represented cooperative games the minimum excess coalition problem can be phrased as a problem of the form (\ref{prob:comb-opt}). Simply take $\cX$ to be the set of incidence vectors of subsets of $[n]$ and take $f(x)$ to be $x(S) - \nu(S)$.

Now consider a linear subspace $\cL \subseteq \R^E$. For our combinatorial optimization problem (\ref{prob:comb-opt}), the associated \emph{linear subspace avoidance problem} is 
\begin{equation}\tag{$P_\cL$}\label{prob:lsa}
    \max\{f(x):x \in \cX\backslash \cL\}
\end{equation}
Even when (\ref{prob:comb-opt}) can be solved in polynomial time with respect to its compact representation and $\cL$ is given through a basis, (\ref{prob:lsa}) can be NP-hard.
\begin{lemma}\label{lemma:np-hard}
    (\ref{prob:lsa}) is NP-hard in general even when (\ref{prob:comb-opt}) can be solved in polynomial time with respect to its compact representation and $\cL$ is given through a basis.
\end{lemma}

\begin{proof}\label{proof:lemma-np-hard}
    We begin by considering the Two Disjoint Directed Paths (TDDP) problem. This well-known NP-complete problem gives as input a directed graph $G$ with two source nodes $s_1,s_2$ and two sink nodes $t_1,t_2$, and the problem is to decide if there exists a pair of arc disjoint paths, one from $s_1$ to $t_1$ and the other from $s_2$ to $t_2$.
    
    Add an arc $(t_1,s_2)$ if it does not already exist. Call the new graph $G'$. Observe that there are arc disjoint $s_1-t_1$ and $s_2-t_2$ paths in $G$ if and only if there is an $s_1-t_2$ path in $G'$ using the arc $(s_2,t_1)$ in $G'$. Let $\cX$ be the set of incidence vectors of $s_1-t_2$ paths in $G'$. The graph $G'$ with $s_1,s_2,t_1,t_2$ labelled may serve as a compact presentation of $\cX$. Let $c\in \R^{E(G')}$ be the all-ones vector. Furthermore, the corresponding problem (\ref{prob:comb-opt}) can be solved in polynomial time with respect to the encoding size of $G'$ and $c$.
    
    Suppose we have an oracle which can solve the corresponding instance of (\ref{prob:lsa}) in polynomial time for any linear subspace $\cL \subseteq\R^{E(G')}$. Consider the particular linear subspace 
    $$\cL:= \{x \in \R^{E(G')}: x_{(t_1,s_2)} = 0\}.$$
    Observe that there an $s_1-t_2$ path in $G'$ using arc $t_1-s_2$ if and only if $\cX\backslash \cL$ is nonempty. Using our oracle for (\ref{prob:lsa}) we can decide in polynomial time if this is the case.
    $\blacksquare$\end{proof}

Observe that when we formulate the minimum excess coalition problem for a cooperative game $(n,\nu)$ as a problem of the form (\ref{prob:comb-opt}) and we take $\cL = \opspan(V_{i-1})$ then (\ref{prob:lsa}) is the ellipsoid method separation problem for (\ref{lp:relaxed-maschler}), the $i$-th linear program in the relaxed MPS Scheme. This discussion yields the following easy lemma

\begin{lemma}\label{lemma:minimum-excess-to-nucleolus}
    If (\ref{prob:comb-opt}) is a minimum excess coalition problem of a cooperative game $(n,\nu)$ and one can solve the associated (\ref{prob:lsa}) for any $\cL$ in polynomial time then the nucleolus of $(n,\nu)$ can be computed in polynomial time.
\end{lemma}

\subsection{Reducing Linear Subspace Avoidance to Congruency-Constrainted Optimization}

\paragraph{}The goal of this subsection is to show the connection between solving (\ref{prob:lsa}) and solving congruency-constrained optimization. This connection was first drawn in the work of Pashkovich~\cite{Pashkovich2018} for the special case of weighted voting games. Here we abstract their work to apply it to our more general framework.

By the following lemma, we can restrict our attention from linear independence over $\R$ to linear independence over finite fields. We present the proof for completeness.

\begin{lemma}\label{lemma:finite-field}
(Pashkovich~\cite{Pashkovich2018}) Let $P$ be a set of prime numbers such that $|P| \geq \log_2(n!)$ with $n\geq 3$. A set of vectors $v_1,\dots, v_k \in \{0,1\}^n$ are linearly independent over $\R$ if and only if there exists $p \in P$ such that $v_1,\dots, v_k$ are linearly independent over $\F_p$.

Moreover, the set $P$ can be found in $O(n^3)$ time, and each $p$ in $P$ can be encoded in $O(\log(n))$ bits.
\end{lemma}

\begin{proof}\label{proof:lemma:finite-field}
    (~\cite{Pashkovich2018})
    Let $A$ be the $n\times k$ matrix whose $i^\text{th}$ column is vector $v_i$. If $v_1,\dots,v_k$ are linearly independent over $\F_p$ then there exists a $k\times k$ submatrix $B$ of $A$ such that $\det(B) \neq 0$ (over $\F_p$). Then $\det(B) \neq 0$ over $\R$ and hence $v_1,\dots, v_k$ are linearly independent over $\R$.
    
    \paragraph{}
    Now suppose that $v_1,\dots, v_k$ are linearly dependent over $\F_p$ for all $p \in P$. If $k > n$ then clearly $v_1,\dots, v_k$ are linearly dependent over $\R$ and we are done. So suppose that $k \leq n$. Let $B$ be a $k\times k$ submatrix of $A$. We will show that $\det(B) = 0$.
    
    For each $p \in P$, $p$ divides $\det(B)$ since $v_1,\dots, v_k$ are linearly dependent over $\F_p$. Note that $\det(B)$ is an integer since it is the determinant of a $0$-$1$ matrix. Since each $p \in P$ is prime, this implies that 
    $$\prod_{p \in P} p \mid \det(B),$$
    and hence, if $\det(B) \neq 0$ then
    $$\prod_{p \in P}p \leq \det(B).$$
    But in this case,
    $$2^{\log_2(n!)} < \prod_{p \in P} p \leq \det(B) \leq n!,$$
    with the last inequality following since $B$ is a $0$-$1$ matrix, yielding a contradiction. Therefore $\det(B) = 0$ as desired.

    Now generate a set $P$ of primes such that $|P| \geq \log_2(n!).$ By the Prime Number Theorem, we can simply find the first $\log_2(n!)$ primes in $O(n^3)$ time and thus we can construct $P$ in time polynomial in $n$. Furthermore, the value of each prime in $P$ will be polynomial in $n$ (i.e.\ each prime in $P$ can be encoded with $O(\log(n))$ bits). 
    $\blacksquare$
\end{proof}

This lemma enables us to reduce the problem (\ref{prob:lsa}) to the problem of computing (\ref{prob:comb-opt}) subject to a congruency constraint with respect to a given prime $p$, $k \in \Z_p$, $v \in Z^E_p$:
\begin{equation}\tag{$P_{\cL,p,v,k}$}\label{prob:lsa-3}
    \max\{f(x):x \in \cX, v^Tx = k \mod p\}.
\end{equation}
\begin{lemma}\label{lemma:subspace-to-congruency}
    If one can solve (\ref{prob:lsa-3}) in time $T$ then one can solve (\ref{prob:lsa}) in time $O(n^6T)$.
\end{lemma}
\begin{proof}\label{proof:lemma:subspace-to-congruency}
    For a linear subspace $\cL \subseteq \R^n$, we say that $x \in \opspan(\cL)$ over $\F_p$ if there exists a basis $B$ of $\cL$ such that $B\cup\{x\}$ is linearly dependent over the field $\F_p$. By Lemma \ref{lemma:finite-field} we can solve (\ref{prob:lsa}) by solving, $$\max\{f(x): x \in \cX, x \not\in\opspan(\cL) \text{ over $\F_p$ for some $p \in P$}\}$$
    The above problem can be solved by solving for each $p \in P$,
    \begin{equation}\tag{$P_{\cL,p}$}\label{prob:lsa-1}
        \max\{f(x): x \in \cX, x \not\in\opspan(\cL)\text{ over $\F_p$}\}
    \end{equation}
    and taking the solution of maximum objective value found. Let $B_p$ be a basis of $\cL^\perp$ over $\F_p$. Assuming $\cL$ is presented to us through a fixed basis, we can compute $B_p$ in polynomial time via Gaussian Elimination~\cite{Edmonds1967}. Now, $x$ is not in the span of $\cL$ over $\F_p$ if and only if there exists $v \in B_p$ such that $v^Tx \neq 0 \mod p.$
    Hence we can solve (\ref{prob:lsa-1}) by solving for each $v \in B_p$,
    \begin{equation}\tag{$P_{\cL,p,v}$}\label{prob:lsa-2}
        \max\{f(x): x \in \cX, v^Tx\neq 0 \mod p\}
    \end{equation}
    and taking the solution of maximum value found. Now $v^Tx \neq 0\mod p$ if and only if there exists $k \in [p-1]$ such that $v^Tx = k$ (over $\F_p$). Thus we can solve (\ref{prob:lsa-2}) by solving for each $k \in [p-1]$, (\ref{prob:lsa-3}) and taking the solution of maximum value found. 
    
    Hence by solving $O(|P|(n-\dim(\cL))\max_{p\in P}p)=O(n^3\max_{p \in P}p) = O(n^6)$ congruency-constrained optimization problems of the form (\ref{prob:lsa-3}) we can solve (\ref{prob:lsa}).
    $\blacksquare$
\end{proof}

\section{Dynamic Programming}\label{sec:dynamic-programming}
\paragraph{}
Our goal is to define a class of problems where tractability of (\ref{prob:comb-opt}) can be lifted to tractability of (\ref{prob:lsa-3}) and hence via Lemma~\ref{lemma:subspace-to-congruency} to (\ref{prob:lsa}). Our candidate will be problems which have a dynamic programming formulation. The model of dynamic programming we propose is based on the model of Martin, Rardin, and Campbell~\cite{martin1990polyhedral}.

The essence of a dynamic programming solution to a problem is a decomposition of a solution to the program into optimal solutions to smaller subproblems. We will use a particular type of hypergraph to describe the structure of dependencies of a problem on its subproblems. 

To begin we will need to introduce some concepts. A \emph{directed hypergraph} $H=(V,E)$ is an ordered pair, where $V$ is a finite set referred to as the \emph{vertices} or \emph{nodes} of the hypergraph, and $E$ is a finite set where each element is of the form $(v,S)$ where $S\subseteq V$ and $v \in V\backslash S$. We refer to the elements of $E$ as \emph{edges} or \emph{arcs} of $H$. For an arc $e=(v,S) \in E$ we call $v$ the \emph{tail} of $e$ and say $e$ is \emph{outgoing} from $v$. We call $S$ the \emph{heads} of $e$, call each $u \in S$ a \emph{head} of $e$, and say $e$ is \emph{incoming} on each $u \in S$. We call vertices with no incoming arcs \emph{sources} and we call vertices with no outgoing arcs \emph{sinks}. For a directed hypergraph $H$, the set $L(H)$ denotes the set of sinks of $H$.

For any non-empty strict subset of vertices $U\subset V$, we define the \emph{cut} induced by $U$, denoted $\delta(U)$, as follows
$$\delta(U) := \{(v,S) \in E: v \in U \text{ and } S\cap (V\backslash U)\neq \emptyset\}.$$
We say a directed hypergraph is \emph{connected} if it has no empty cuts.

A \emph{directed hyperpath} is a directed hypergraph $P$ satisfying the following:
\begin{itemize}
\item there is a unique vertex $s\in V(P)$ identified as the \emph{start} of $P$,
\item the start $s$ is the tail of at most one arc of $P$, and the head of no arcs of $H$,
\item every vertex in $V(P)\backslash \{s\}$ is the tail of precisely one arc of $H$,
\item $P$ is connected.
\end{itemize}
Observe that there is at least one, and potentially many, vertices of a path which have one incoming arc and no outgoing arcs. These vertices we call the \emph{ends} of the path. If there is a path starting from a vertex $u$ and ending with a vertex $v$ then we say $u$ is an \emph{ancestor} to $v$ and $v$ is a \emph{descendant} of $u$. For any vertex $v \in V(H)$, the subgraph of $H$ \emph{rooted at} $v$, denoted $H_v$, is the subgraph of $H$ induced by the descendants of $v$ (including $v$).

We say that a directed hypergraph $H=(V,E)$ is \emph{acyclic} if there exists a topological ordering of the vertices of $H$. That is to say, there exists a bijection $t:V\rightarrow [\abs{V}]$ such that for every $(v,S) \in E$, for each $u \in S$,
$t(v) < t(u).$

A common approach to dynamic programming involves a table of subproblems (containing information pertaining to their optimal solutions), and a recursive function describing how to compute an entry in the table based on the values of table entries which correspond to smaller subproblems. The values in the table are then determined in a bottom-up fashion. In our formal model, the entries in the table correspond to vertices of the hypergraph, and each hyperarc $(v,S)$ describes a potential way of computing a feasible solution to the subproblem at $v$ by composing the solutions to the subproblems at each node of $S$.

Consider a problem of the form (\ref{prob:comb-opt}). That is, we have a feasible region $\cX \subseteq \R^n$ and an objective function $f:\cX \rightarrow \R$ and we hope to maximize $f(x)$ subject to $x \in \cX$. We need some language to describe how solutions to the dynamic program, i.e.\ paths in the directed hypergraph, will map back to solutions in the original problem space. To do this mapping back to the original space we will use an affine function. A function $g: \R^m \rightarrow \R^n$ is said to be \emph{affine} if there exists a matrix $A \in \R^{n\times m}$ and a vector $b \in \R^n$ such that for any $x \in \R^m$, $g(x) = Ax + b$. 

Oftentimes an affine function $g$ will have a domain $R^E$ indexed by a finite set $E$. When this happens for any $S \subseteq E$ we use $g(S)$ as a shorthand for $g(\chi(S))$ where $\chi(S)$ is the incidence vector of $S$. We further shorten $g(\{e\})$ to $g(e)$.

\begin{definition}Let $H=(V,E)$ be a directed acyclic connected hypergraph  with set of sources $T$. Let $\cP(H)$ denote the set of paths in $H$ which begin at a source in $T$ and end only at sinks of $H$. Let $g: \R^E\rightarrow \R^{n}$ be an affine map which we will use to map between paths in $\cP(H)$ and feasible solutions in $\cX$. Let $c:\R^E \rightarrow \R$ be an affine function we will use as an objective function. We say $(H,g,c)$ is a \emph{dynamic programming formulation} for (\ref{prob:comb-opt}) if $g(\cP(H)) = \cX$, and moreover for any $x \in \cX$, 
$$ f(x) = \max_{P \in g^{-1}(x)} c(P).$$
\end{definition}
In other words, the optimal values of  
\begin{equation}\tag{DP}\label{prob:dp-paths}
    \max\{c(P): P \in \cP(H)\}
\end{equation}
and (\ref{prob:comb-opt}) are equal, and the feasible region of (\ref{prob:comb-opt}) is the image (under $g$) of the feasible region of (\ref{prob:dp-paths}). The \emph{size} of a dynamic programming formulation is the number of arcs in $E(H)$.

In~\cite{martin1990polyhedral} the authors show that (\ref{prob:dp-paths}) has a totally dual integral extended formulation of polynomial size. Thus they show that (\ref{prob:dp-paths}) can be solved in polynomial time via linear programming. They further show that the extreme point optimal solution of this extended formulation lies in $\{0,1\}^{E}$ under the following \emph{reference subsets} condition: there exists a ground set $I$, and nonempty subsets $I_v\subseteq I$ for each vertex $v \in V(H)$ satisfying
\begin{enumerate}
\item $I_j \subseteq I_\ell$ for all $(\ell, J) \in H$ such that $j \in J$
\item and $I_j \cap I_{j'} = \emptyset$ for all $(\ell, J) \in E(H)$ such that $j,j' \in J$ with $j\neq j'$.
\end{enumerate}

This condition is equivalent to the following \emph{no common descendants} condition: for each $(\ell, J) \in E(H)$ for all $u\neq v \in J$, there does not exist $w \in V(H)$ such that $w$ is a descendant of both $u$ and $v$.

\begin{lemma}\label{lemma:no-common-descendants}
    For any directed acyclic hypergraph $H=(V,E)$ the reference subsets condition is equivalent to the ``no common descendants'' condition defined above.
\end{lemma}

\begin{proof}\label{proof:lemma:no-common-descendants}
    First suppose that $H$ does not satisfy the \emph{no common descendants} condition. Then there exists $(\ell, J) \in E$ such that there exist $u, v \in J$ and $w \in V$ such $w$ is a descendant of $u$ and of $v$. Suppose for a contradiction that $H$ has a \emph{reference subset} system with ground set $I$. 
    
    We claim that $I_w \subseteq I_v$. The proof of this claim will symmetrically show that $I_w \subseteq I_u$. Then $I_w \cap I_v \supset I_w \neq \emptyset$ violating the second property of a reference subsset system.
    
    To prove the claim we will prove something stronger. In particular we will show that for any $x,y \in V$ such that $y$ is a descendant of $x$, we have that $I_y\subseteq I_x$. Suppose not. Choose a counterexample $x,y$ with path $P$ starting at $x$ and ending at $y$ so that the number of edges in $P$ is minimal. Clearly $|E(P)|\neq 0$ as otherwise $x=y$. Now, from the definition of $P$ there exists an arc $(x,J) \in E(P)$ and there exists $z \in J$ such that there is a subgraph of $P$, denoted $P'$, such that $P'$ is a path starting at $z$ and ending at $y$. By minimality, $I_y \subseteq I_z$. By the first property of reference subset systems, $I_z \subseteq I_x$. Thus $I_y \subseteq I_x$ contradicting that $x,y$ and $P$ form a counterexample.
    
    Now for the other direction of the equivalence suppose that $H$ satisfies the \emph{no common descendants} condition. We will construct a reference subset system for $H$ as follows. Let $I = V$ and for each $v \in V$ let $I_v$ be the set of descendants of $v$. Then $I$ satisfies the first property of a reference subset system since the descendant relation is transitive. Further, the \emph{no common descendants} condition implies that $I$ satisfies the second property of a reference subset system. Lastly, no $I_v$ is empty since every vertex is their own descendant.
    $\blacksquare$
\end{proof}

We say that a dynamic programming formulation $(H,g,c)$ of a problem (\ref{prob:comb-opt}) is integral if $H$ satisfies the no common descendants condition. By the preceding discussion we have the following lemma
\begin{lemma}\label{lemma:polytime-dp}
    If a problem (\ref{prob:comb-opt}) has an integral dynamic programming formulation $(H,g,c)$ then (\ref{prob:comb-opt}) can be solved in time polynomial in the encoding of $(H,g,c)$.
\end{lemma}

\subsection{Congruency Constrained Dynamic Programming}\label{sec:congruency-constrained-dp}

In this subsection our goal is to show that when a problem of the form (\ref{prob:comb-opt}) has a dynamic programming formulation, then its congruency constrained version (\ref{prob:lsa-3}) has a dynamic programming formulation that is only a $O(p^3)$  factor larger than the formulation for the original problem. This will prove Theorem~\ref{th:congruency}.

We begin with a handy lemma for constructing dynamic programming formulations of combinatorial optimization problems.
\begin{lemma}\label{lemma:composition}
    If $(H,g,c)$ is a dynamic programming formulation for (\ref{prob:comb-opt}) and $(H',g',c')$ is a dynamic programming formulation for (\ref{prob:dp-paths}) with respect to hypergraph $H$ and costs $c$ then $(H', g\circ g', c')$ is a dynamic programming formulation for (\ref{prob:comb-opt}).
\end{lemma}

\begin{proof}\label{proof:lemma:composition}
    The function $g\circ g'$ is map between $\cP(H')$ and $\cX$, and moreover it is affine since both $g$ and $g'$ are affine. Furthermore,
    $$g\circ g'(\cP(H')) = g(\cP(H)) = \cX.$$
    Finally, for any $P \in \cP(H')$,
    $$ f(g\circ g'(P)) = c(g'(P)) = c'(P).$$
$\blacksquare$
\end{proof}

\paragraph{}
Consider a directed hypergraph $H=(V,E)$ and an edge $(u,S) \in E$. For $v \in S$ we define the \emph{hypergraph obtained from the subdivision of }$(u,S)$ \emph{with respect to }$v$ to be the hypergraph $H'=(V',E')$ where $V' = V \dot\cup \{b_v\}$ for a new dummy vertex $b_v$ and
$$E' = \left(E\backslash\{(u,S)\}\right) \cup \{(u,\{v,b_v\}), (b_v, S\backslash\{v\})\}.$$
That is, $H'$ is obtained from $H$ by replacing edge $(u,S)$ with two edges: $(u, \{v,b_v\})$ and $(b_v, S\backslash\{v\})$. We call the edges $(u, \{v,b_v\})$ and $(b_v, S\backslash\{v\})$ the \emph{subdivision} of edge $(u,S)$.

\begin{lemma}\label{lemma:subdivision}
    Let $H=(V,E)$ be a directed acyclic hypergraph and let $H'=(V',E')$ be the directed acyclic hypergraph obtained via a subdivision of $(u,S) \in E$ with respect to $v\in S$. Then there is an affine function $g: \R^{E'}\rightarrow \R^E$, such that for any affine function $c:\R^E\rightarrow\R$, there exists an affine function $c':\R^{E'}\rightarrow \R$ such that $(H,g,c')$ is a dynamic programming formulation of the problem (\ref{prob:dp-paths}) on $H$ with objective $c$.

    Moreover if $H$ satisfies the ``no common descendants'' property, this dynamic programming formulation is integral.
\end{lemma}

\begin{proof}\label{proof:lemma:subdivision}
    Take $g$ to be the affine function defined as follows. For any $e \in E$ and $x \in \R^{E'}$,
    $$g(x)_e := \begin{cases}
        x_{(u,\{v, b_v\})}, &\text{if $e=(u,S)$}\\
        x_e, &\text{otherwise.}
    \end{cases}$$
    It is not hard to see that $g$ is affine, and in fact linear. In fact, $g$ is simply the function which identifies $e$ with $(u,\{v,b_v\})$ and every other edge in $e$ with the corresponding edge in $E'$.

    To see that $g$ is a bijection between $\cP(H)$ and $\cP(H')$ it suffices to observe that a path $P$ in $\cP(H')$ uses $(u,\{v,b_v\})$ if and only $P$ uses $(b_v, S\backslash \{v\})$. This follows immediately from the definition of hyperpaths. 

    Now for any $c:\R^E\rightarrow \R$ let $c':\R^{E'}\rightarrow \R$ be the unique affine function which acts on the standard basis in the following way: for any $e \in E'$,
    $$c'(e) = \begin{cases}
        c((u,S)), &\text{ if $e = (u,\{v,b_v\})$} \\
        0, &\text{ if $e = (b_v, S\backslash\{v\})$} \\
        c(e), &\text{otherwise.}
    \end{cases}$$
    Again since a path $P$ in $\cP(H')$ uses $(u,\{v,b_v\})$ if and only $P$ uses $(b_v, S\backslash \{v\})$, it is easy to see that $c(g(P) = c'(P)$ for any $P \in \cP(H')$. Hence $(H',g,c')$ is the desired dynamic programming formulation.

    Since the subdivision operation preserves the ``no common descendants'' property, $(H',g,c)$ is integral if $H$ satisfies the no common descendants condition.
$\blacksquare$\end{proof}

 For a directed hypergraph $H=(V,E)$ let $\Delta(H) := \max\{|S|: (u,S) \in E\}$ and let $\Gamma(H) := \abs{\{(u,S) \in E: |S| = \Delta(H)\}}$. The following Lemma shows that we may assume the number of heads of any arc in a dynamic programming formulation is constant.

\begin{lemma}\label{lemma:constant-arcs}
     Consider a combinatorial optimization problem of the form (\ref{prob:comb-opt}). If there exists a dynamic programming formulation $(H,g,c)$ for (\ref{prob:comb-opt}) then there exists a dynamic programming formulation $(H^*,g^*,c^*)$ for (\ref{prob:comb-opt}) such that $\Delta(H^*) \leq 2$, and $|E(H^*)| = \sum_{u \in V(H)}\sum_{(u,S) \in E(H)}(|S|-1).$

     Moreover, if $H$ is integral then $H^*$ is integral.
\end{lemma}

\begin{proof}\label{proof:lemma:constant-arcs}
    We proceed by double induction on $\Delta(H)$ and $\Gamma(H)$.
    In the base case, for any $\Gamma(H)$ if $\Delta(H) \leq 2$ then $(H,g,c)$ is the desired dynamic program. In the inductive case consider $H$ such that $\Delta(H) > 2$ and suppose the lemma holds on all directed acyclic hypergraphs $H'$ with $\Delta(H') < \Delta(H)$ or $\Delta(H') = \Delta(H)$ and $\Gamma(H') < \Gamma(H)$. Since $\Delta(H) \geq 3$ there exists an edge $(u,S) \in E(H)$ such that $|S| \geq 3$. Let $(H',g',c')$ be the dynamic program for (\ref{prob:dp-paths}) on $H$ with objective $c$ given by Lemma~\ref{lemma:subdivision} where $H'$ is obtained by a subdivision of $(u,S)$ with respect to vertex $v \in S$. By Lemma \ref{lemma:composition} $(H', g\circ g', c')$ is a dynamic programming formulation for (\ref{prob:comb-opt}). 
    
    Now notice that either $\Delta(H') = \Delta(H) -1$ or $\Delta(H') = \Delta(H)$ and $\Gamma(H') = \Gamma(H)-1$. Hence by induction there is a dynamic program $(H^*, g^*, c^*)$ for the problem (\ref{prob:dp-paths}) on hypergraph $H'$ with respect to objective $c'$ such that $\Delta(H^*)\leq 2$ and 
\begin{align*}|E(H^*)| &= \sum_{u' \in V(H')}\sum_{(u',S') \in E(H')} (|S'|-1)\\
&= |S\backslash\{v\}| -1  + |\{v, b_v\}|-1 + \sum_{u' \in V(H)}\sum_{(u',S') \in E(H)\backslash\{(u,S)\}} |S'|-1 \\
&= \sum_{u' \in V(H)}\sum_{(u',S') \in E(H)}|S'|-1,
\end{align*}
where the second equality follows since $H'$ is obtained from $H$ via a subdivision of $(u,S)$ with respect to $v$.
By Lemma \ref{lemma:composition} $(H^*, g\circ g' \circ g^*, c^*)$ is a dynamic programming formulation for (\ref{prob:comb-opt}) as desired.

Since subdivision preserves the ``no common descendants'' property, if $H$ is integral then $H^*$ is integral.
$\blacksquare$
\end{proof}

The next lemma is our main techincal lemma. It provides the backbone of our dynamic programming formulation for (\ref{prob:lsa-3}) by showing that we track the congruency of all hyperpaths rooted at a particular vertex by expanding the size of our hypergraph by a factor of $p^{\Delta(H)}+1$.

\begin{lemma}\label{lemma:congruency}
Let $H=(V,E)$ be a directed acyclic hypergraph. Let $p$ be a prime. Let $k \in \Z_p$ and let $a \in \Z^E_p$. There exists a directed acyclic hypergraph $H' = (V',E')$ and an affine function $g': \cP(H')\rightarrow \cP(H)$, $g'(x) = Ax + b$, such that:
\begin{enumerate}[1)]
    \item $\abs{E'} \leq p^{\Delta(H)+1}\abs{E}$
    \item For every $v \in V\backslash L(H)$, for every $k \in \Z_p$, if $\{P \in \cP(H_v): a(P) = k \mod p\} \neq \emptyset$ then there exists $v' \in V(H')$ such that $$g'(\cP(H'_{v'}))=\{P \in \cP(H_v) : a(P) = k \mod p\}.$$
\end{enumerate}
Moreover if $H$ satisfies the ``no common descendants" property then $H'$ satisfies the ``no common descendants" property.
\end{lemma}

\begin{proof}\label{proof:lemma:congruency}
    We may assume that each source $u$ of $H$ has at most one arc outgoing from $u$. To see this, if the arcs leaving $u$ are $(u,S_1),\dots, (u, S_\ell)$ then we can replace $u$ with $\ell$ vertices $u_1, \dots, u_\ell$ and replace the arcs of $u$ with $(u_1,S_l)$,\dots, $(u_\ell, S_\ell)$. This does not increase the number of arcs, and preserves the structure of the paths in $\cP(H)$.
    
    Now suppose for a contradiction that the lemma is false. Consider a counterexample $H=(V,E)$ which minimizes $|E|$. Then $E\neq \emptyset$, otherwise $H$ is trivially not a counterexample. Let $u\in V$ be a source of $H$ which has an outgoing arc, let $(u,S) \in E$ be the arc outgoing from $u$ and let $d = |S|$. Arbitrarily fix an indexing on the vertices of $S$ as in $S=\{s_1,\dots, s_d\}$.
    
    Since $H-u$ has fewer arcs than $H$, it is not a counterexample to the lemma. Hence there exists a directed acyclic hypergraph $\overline{H}$ and an affine function $\overline{g}$ satisfying properties $1)$ and $2)$ of the lemma with respect to $H-u$.

    We now construct a new directed acylic hypergraph $H'$ and affine function $g': \R^{E'}\rightarrow \R^{E}$ via the following procedure.
    \begin{enumerate}
        \item Initialize $H'=(V',E')$ to be the hypergraph $\bar{H}$.
        \item For every $k \in \Z_p$ such that $\{P \in \cP(H): a(P) = k \mod p\} \neq \emptyset$ do:
        \begin{enumerate}
            \item Add a new vertex $u^k$ to $V'$.
            \item For each $q \in \Z_p^d$ such that $a^Tq = k- g((u,S))\mod p$ do:
            \begin{itemize}
                \item if for every $i \in [d]$ there exists $s'_i \in V'$ such that $g$ is a bijection between $\cP(H'_{s'_i})$ and $\{P \in \cP(H_{s_i}): a(P) = k-g((u,S)) \mod p \}$ then add arc $(u^k, \bigcup_{i \in [d]}s'_i)$ to $E'$.
            \end{itemize}
        \end{enumerate}
        \item We will define $g'$ by its action on the standard basis. For each $e \in E'$ we define $$g'(e) := \begin{cases}
            \chi((u,S)), &\text{ if $e$ is outgoing from $u^k$ for some $k \in \Z_p$}\\
            \overline{g}(e), &\text{ otherwise.} 
        \end{cases}$$
    \end{enumerate}
    We will first verify that property $1)$ is satisfied by $H'$. Since $\overline{H}$ satisfies property $1)$, $\abs{E(\overline{H})} \leq p^{\Delta(H-v)+1} (\abs{E}-1)$. By our construction of $H'$, $E'$ has at most $p\cdot \abs{Z^d_p}$ more arcs than $E(\overline{H})$, and hence
    $$\abs{E'} \leq \abs{\overline{E}} + p\cdot \abs{Z^d_p} \leq p^{\Delta(H-v)+ 1} (\abs{E}-1) + p^{d+1} \leq p^{\Delta(H) + 1}\abs{E}.$$
    Now we will show that property $2)$ is satisfied by $H'$ and $g'$. Let $v \in V\backslash L(H)$. Let $k \in \Z_p$ and suppose that $\{P \in \cP(H_v) : a(P) = k \mod p\} \neq \emptyset$. If $v \neq u$, then since $\overline{H}$ and $\overline{g}$ satisfies property $2)$, $H'$ and $g'$ satisfy property $2)$ by steps $1.$ and $3.$ of the construction. Thus we may assume that $v = u$. Then the vertex $u^k$ was added to $V'$ in step $2a)$ of the construction. It suffices to show that $g'$ is a bijection between $\cP(H'_{u^k})$ and $\{P\in\cP(H_u) : a(P) = k \mod p\}$.

    First we show surjectivity. Let $P \in \{P \in \cP(H_u): a(P)=k \mod p\}$. Then $P$ is the disjoint union of $(u,S)$ and $d$ paths $P_1, \dots, P_d$ where each $P_i$ is in $\cP(H_{s_i})$ respectively. For each $i \in [d]$, let $k_i = a(P_i)\mod p$. By property $2)$ applied to $\overline{H}$ and $\overline{g}$, and by our construction of $H'$ starting from $\overline{H}$ and $g'$ starting from $\overline{g}$, for each $i\in [d]$ there exists a vertex $s'_i \in V'\cap \overline{V}$ such that $g'$ is a bijection between $\cP(H'_{s'_i})$ and $\{P \in \cP(H_{s_i}) : a(P) = k_i \mod p\}$. Since $g'$ is such a bijection there exists $P'_i \in \cP(H'_{s'_i})$ such that $g'(P'_i) = P_i$.
    
    By step $2b)$ of our construction, there is an arc of the form $(u^k, \bigcup_{i \in [d]} s'_i)$ in $E'$. Hence there is a path $P' \in \cP(H'_{u^k})$ which is a disjoint union of $(u^k, \bigcup_{i \in [d]} s'_i)$ with $P'_1,\dots, P'_d$. But then
    $$g'(P') = g'(u^l,\bigcup_{i \in [d]} s'_i) \bigcup_{i\in[d]}g'(P'_d) = (u,S) \bigcup_{i \in [d]}P_i = P.$$

    Now we show injectivity. Let $P', Q' \in \cP(H'_{u^k})$ be paths in $H'$ such that $g'(P') = g'(Q')$. Let $P = g'(P') = g'(Q')$. We consider the disjoint union of hyperpaths obtained by deleting $u^k$ (and hence the single arc outgoing from $u^k$) from $P'$ and $Q'$. Observe that 
    $$P - u = g'(P'-u) = g'(Q'-u).$$
    By the bijectivity of $g'$ inherited from $\overline{g}$, this implies that $P'-u = Q'-u$. But that necessarily implies that $P'$ and $Q'$ have the same arc outgoing from $u^k$. Thus $P' = Q'$. Therefore we have shown property $2)$.

    It is easy to see that the construction preserves the no common descendants property.
$\blacksquare$
\end{proof}

We are now ready to show our main theorem, which says that for any combinatorial optimization problem which has a dynamic program, its congruency-constrained version also has a dynamic program of proportional size.

\begin{theorem}\label{th:congruency}
Consider an instance of a combinatorial optimization problem (\ref{prob:comb-opt}). Let $p$ be a prime, let $v\in \Z_p^n$, and let $k \in \Z_p$. Consider the corresponding congruency-constrained optimization problem (\ref{prob:lsa-3}). If (\ref{prob:comb-opt}) has a dynamic programming formulation $(H,g,c)$ then (\ref{prob:lsa-3}) has a dynamic programming formulation $(H',g',c')$ such that $\abs{E(H')} \leq p^3\cdot\abs{V(H)}\cdot\abs{E(H)}$.

Moreover if $(H,g,c)$ is integral then $(H',g',c')$ is integral.
\end{theorem}

\begin{proof}\label{proof:th:congruency}
    We first apply Lemma~\ref{lemma:constant-arcs} to problem (\ref{prob:comb-opt}) and dynamic programming formulation $(H,g,c)$ to obtain a dynamic programming formulation $(H^*, g^*, c^*)$ for (\ref{prob:comb-opt}) such that $$\Delta(H^*) \leq 2\quad \text{and}\quad\abs{E(H^*)} \leq \abs{V(H)}\cdot\abs{E(H)}.$$
    
    Since $g^*$ is an affine function, there exists a matrix $A$ and a vector $b$ such that $g^*(x) = Ax + b$ for any $x$. Now apply Lemma~\ref{lemma:congruency} to hypergraph $H^*$ with prime $p$, integer $k-v^Tb \in \Z_p$, and  vector $A^Tv \in \Z^E_p$. We obtain a directed acyclic hypergraph $H'$ and affine function $g' : \cP(H')\rightarrow \cP(H^*)$ such that 
    \begin{enumerate}[1)]
        \item $\abs{E'}\leq p^{\Delta(H^*) + 1}|E(H^*)|\leq p^3\abs{V(H)}\abs{E(H)}$ and
        \item For every $v \in V(H^*)\backslash L(H^*)$, if $$\{P \in \cP(H^*_v) : v^TA\chi(P) = k-v^Tb \mod p\} \neq \emptyset$$ then there exists $v' \in V(H')$ such that $g'$ is a bijection between $\cP(H'_{v'})$ and $\{P \in \cP(H^*_v) : v^TA\chi(P)=k-v^Tb \mod p\}$.
    \end{enumerate}
    Let $c': \R^{E(H')} \rightarrow \R$ be the affine function $c^*\circ g'$. Let $T$ be the set of sources in $H'$ that satisfy the hypothesis in property $2)$. We claim that 
    $$D:=(\bigcup_{u \in T}H'_u, g^*\circ g', c')$$ is a dynamic programming formulation for (\ref{prob:comb-opt}). Indeed by property $2)$, $$g'(\cP(\bigcup_{u \in T}H'_u))=\{P \in : \cP(H^*) : v^TA(P) = k-v^Tb \mod p\}.$$ Now observe that 
    \begin{align*}
        &\{P \in : \cP(H^*) : v^TA(P) = k-v^Tb \mod p\} \\
        &= \{P \in : \cP(H^*) : v^T(A(P)+b) = k \mod p\}\\
         &= \{P \in : \cP(H^*) : v^Tg^*(P)= k\mod p\}.
    \end{align*}
    Since $g^*(\cP(H^*))=\cX$, we have that $$g^*(\{P \in : \cP(H^*) : v^Tg^*(P)= k\mod p\})=\{x \in \cX: v^Tx = k \mod p\}.$$ Therefore $$g^*\circ g'(\cP(\bigcup_{u \in T}H'_u))=\{x \in \cX: v^Tx = k \mod p\},$$ the feasible region of (\ref{prob:lsa-3}).

    Lastly for any $P \in \cP(\bigcup_{u \in T}H'_u)$, we have
    $$f(g^*\circ g' (P)) = c(g'(P)) = c'(P).$$
    Thus $D$ is a dynamic programming formulation for (\ref{prob:lsa-3}). Observe that the lemmas applied in the construction of $D$ all preserve integrality. Hence if $(H,g,c)$ is integral then $D$ is integral.
$\blacksquare$\end{proof}

Due to Theorem~\ref{th:congruency} and Lemma~\ref{lemma:polytime-dp} we have the following Corollary. 

\begin{corollary}
    If (\ref{prob:comb-opt}) has an integral dynamic programming formulation $(H,g,c)$ then for any $v,k,p$ problem (\ref{prob:lsa-3}) can be solved in time polynomial in size of $H$, the prime $p$, and the encoding of $g,c,v,k$,
\end{corollary}

Via this Corollary, Lemma~\ref{lemma:minimum-excess-to-nucleolus}, and Lemma~\ref{lemma:subspace-to-congruency} we obtain our first main result: Theorem~\ref{th:dp-nucleolus}.

\section{Applications}\label{sec:applications}
\paragraph{}
In this section we show a couple of applications of Theorem~\ref{th:dp-nucleolus} to computing the nucleolus of cooperative games. The first application is to Weighted Voting Games. In~\cite{Pashkovich2018} a pseudopolynomial time algorithm for computing the nucleolus of Weighted Voting Games was given. We show how the same result can be obtained as a special case of Theorem~\ref{th:dp-nucleolus}. Recall that a weighted voting game $(n,\nu)$ has value function $\nu: 2^{[n]} \rightarrow \{0,1\}$ determined by a vector $w \in \Z^n$ and $T \in Z$, such that for any $S \subseteq [n]$, $\nu(S) =1$ if and only if $w(S) \geq T$.

We partition $2^{[n]}$ into two classes: $N_0 := \{S \subseteq [n]: w(S) < T\}$ and $N_1:=\{S \subseteq [n]: w(S) \geq T\}$. If we can design a dynamic programming formulation for the minimum excess coalition problem restricted to $N_0$:
$
    \max\{-x(S): w(S) \leq T-1, S\subseteq [n]\}
$
and a dynamic programming formulation for the minimum excess coalition problem restricted to $N_1$:
$
    \max\{-x(S) + 1:w(S)\ \geq T, S \subseteq [n]\},
$
 then the dynamic programming formulation which takes the maximum of these two formulations will provide a dynamic programming formulation for the minimum excess coalition problem of the weighted voting game.

If we let $W[k, D]$ denote $\max\{-x(S): w(S) \leq D, S\subseteq [k]\}$ then we can solve the minimum excess coalition problem restricted to $N_0$ by computing $W[n,T-1]$ via the following recursive expression, which is essentially a dynamic program for Knapsack Cover,
$$W[k,D] = \begin{cases}
    \max\{W[k-1,D-w_k], W[k-1,D]\}, &\text{ if $k > 1$} \\
    -x_1, &\text{ if $k=1$ and $w_1 \leq D$}\\
    -\infty, &\text{ if $k=1$ and $w_1 > D$}.
\end{cases}$$
It is not hard to construct a dynamic programming formulation $(H_0,g_0,c_0)$ for the minimum excess coalition problem restricted to $N_0$ by following this recursive expression. The hypergraph $H_0$ will in fact be a rooted tree (i.e.\ all heads will have size one), and $H_0$ will have $O(nT)$ vertices and arcs. Via a similar technique, a dynamic programming formulation $(H_1,g_1,c_1)$ with $O(nT)$ arcs can be constructed for the minimum excess problem restricted to $N_1$. Then by taking the union these dynamic programming formulations, we obtain an integral dynamic programming formulation of size $O(nT)$. Therefore by Theorem~\ref{th:dp-nucleolus} we obtain a short proof that 
\begin{theorem}
    (~\cite{Pashkovich2018}~\cite{Elkind2009a}) The nucleolus of a weighted voting game can be computed in pseudopolynomial time.
\end{theorem}

In the following subsections we will see how the added power of hyperarcs lets us solve the more complex problem of computing the nucleolus of $b$-matching games on graphs of bounded treewidth.

\subsection{Treewidth}
\paragraph{}
Consider a graph $G=(V,E)$. We call a pair $(T,\cB)$ a \emph{tree decomposition}~\cite{halin1976s}~\cite{robertson1984graph} of $G$ if $T = (V_T, E_T)$ is a tree and $\cB = \{B_i \subseteq V: i\in V_T\}$ is a collection of subsets of $V$, called \emph{bags}, such that 
\begin{enumerate}
\item $\bigcup_{i \in V_T} B_i = V$, i.e.\ every vertex is in some bag,
\item for each $v \in V$, the subgraph of $T$ induced by $\{i \in V_T: v \in B_i\}$ is a tree, and
\item for each $uv \in E$, there exists $i \in V_T$ such that $u,v \in B_i$.
\end{enumerate}
The \emph{width} of a tree decomposition is the size of the largest bag minus one, i.e.\ 
$\max_{i \in V_T} \{ |B_i| -1 \}.$
The \emph{treewidth} of graph $G$, denoted $\tw(G)$, is minimum width of a tree decomposition of $G$.

We may assume that tree decompositions of a graph have a special structure. We say a tree decompostion $(T,\cB)$ of $G$ is \emph{nice} if there exists a vertex $r \in V_T$ such that if we view $T$ as a tree rooted at $r$ then every vertex $i \in V_T$ is one of the following types:
\begin{itemize}
    \item \textbf{Leaf:} $i$ has no children and $|B_i| = 1$.
    \item \textbf{Introduce:} $i$ has one child $j$ and $B_i = B_j \dot\cup\{v\}$ for some vertex $v \in V$.
    \item \textbf{Forget:} $i$ has one child $j$ and $B_i \dot\cup\{v\} = B_j$ for some vertex $v \in V$.
    \item \textbf{Join:} $i$ has two children $j_1, j_2$ with $B_i = B_{j_1}=B_{j_2}$.
\end{itemize}

It turns out that if a graph has a tree decomposition of width $w$ then a nice tree decomposition of width $w$ times the number of vertices of the graph can be computed in time polynomial in $w$ and the number of vertices of the graph.

\begin{theorem}(~\cite{kloks1994treewidth} Lemma 13.1.3)
  If $G=(V,E)$ has a tree decompostion of width $w$ with $n$ tree vertices then there exists a nice tree decomposition of $G$ of width $w$ and $O(|V|)$ tree vertices which can be computed in $O(|V|)$ time.\end{theorem}

\subsection{Dynamic Program for $b$-Matching Games}

We want to show that on graphs of bounded treewidth, the nucleolus of $b$-matching games can be computed efficiently. Fix a graph $G=(V,E)$, a vector of $b$-values $b \in \Z^V$, and tree decomposition $(T,\cB)$ of treewidth $w$, where $T$ is rooted at $r$, to be used throughout this section. For $i \in V(T)$ let $T_i$ denote the subtree of $T$ rooted at $i$, and also let $G_i := G[\bigcup_{j \in V(T_i)}B_j]$. For any $v \in V(G_i)$, let $\delta_i(v) := \{uv \in E(G_i)\}$.

\paragraph{}
For any $i \in V(T)$, $X \subseteq B_i$, $d \in \{d \in \Z^{B_i}: 0 \leq d \leq \Delta(G)\}$, and $F \subseteq E(B_i)$, we define the combinatorial optimization problem~\ref{prob:min-excess-b-matching-specific} to be the problem of finding a $b$-matching $M$ and a set of vertices $S$ such that $M$ uses only edges of $G_i$, $S$ uses only vertices of $G_i$, the intersection of $M$ and $E(B_i)$ is $F$, the number of edges in $M$ adjacent to $u$ is $d_u$ for each $u$ in $B_i$, and the vertices in $S$ not intersecting an edge in $F$ is $X$. Formally~\ref{prob:min-excess-b-matching-specific} is defined as follows

\begin{align*}\label{prob:min-excess-b-matching-specific}\tag{C[i,X,d,F]}
    \max\ &w(M) - x(S) \\
    \text{s.t. } &|M \cap \delta_i(v) | \leq b_v &\forall v \in V \\
    &|M\cap \delta_i(u)| \leq d_u &\forall u \in B_i \\
    &d_u = 0 &\forall u \in X\\
    &M \cap E(B_i) = F \\
    &S = V(M)\dot\cup X\\
    &X \subseteq B_i\\
    &M\subseteq E(G_i).
\end{align*}

We define~\ref{prob:min-excess-b-union} to be the union over all $C[i,X,d,F]$.

\begin{align*}\label{prob:min-excess-b-union}\tag{C[i]}
    \max\ &w(M) - x(S) \\
    \text{s.t. } (M,S) &\text{ is feasible for }C[i,X,d,F]\\&\text{ for some }(X,d,F)\in B_i\times\Z^{B_i}\times E(G_i).
    \end{align*}

We will show a dynamic programming formulation $(H,g,c)$ for $C[i]$. Since the feasible region of the minimum excess coalition problem for $b$-matching games is the image of the feasible region of $C[i]$ under the linear map which projects out $M$, and $\nu(S) - x(S) = \max_{(M,S) \text{feasible for $C[i]$}} w(M) - x(S),$ the existence of $(H,g,c)$ will imply the existence a dynamic programming formulation of the minimum excess coalition problem for $b$-matching games of the same encoding length.

\begin{lemma}\label{lemma:min-excess-b-matching-specific-dp}
 Let $i \in V(T)$. There exists an integral dynamic programming formulation $(H,g,c)$ for \ref{prob:min-excess-b-union} such that 
\begin{enumerate}[1)]
    \item $\abs{E(H)} \leq \abs{V(T_i)}\cdot w \cdot \Delta(G)^{w} \cdot w^2$ and
    \item For every $j \in V(T_i)$, $X \subseteq B_i$, $d\in \Z^{B_i}$, and $F\subseteq E(B_i)$, if \ref{prob:min-excess-b-matching-specific}) has a feasible solution then there exists $a \in V(H)$ such that $(H_a,g,c)$ is an integral dynamic programming formulation for \ref{prob:min-excess-b-matching-specific}
\end{enumerate}
\end{lemma}

\begin{proof}\label{proof:lemma:min-excess-b-matching-specific-dp}
    We proceed by induction on the distance from $i$ to a leaf of $T$. If $i$ is a leaf of $T$ the Lemma is trivial. So we suppose that $i$ is not a leaf of $T$ and the Lemma holds for all vertices $j$ of $T$ which are closer to a leaf than $i$.

    Now we proceed by case distinction on which class of node $i$ is: an Introduce node, a Forget node, or a Join node.

    \paragraph{Case: Introduce.} If $i$ is an Introduce node then $i$ has one child $j$, and $B_i = B_j \dot\cup\{v\}$ for some vertex $v \in V$. Let $(H',g',c')$ be the dynamic programming formulation for $C[j]$ guaranteed by the inductive hypothesis. We construct a new dynamic programming formulation $(H,g,c)$ via the following procedure:
    \begin{enumerate}
        \item Initialize $H$ to $H'$.
        \item For each $e' \in E(H')$, define the action of $g$ on $e'$ to be $g(e') = g'(e')$.
        \item For each $e' \in E(H')$, define the action of $c$ on $e'$ to be $c(e') = c'(e')$.
        \item For every $X \subseteq B_i$, $d \in Z^{B_i}$, $F \subseteq E(B_i)$ such that $C[i,X,d,F]$ has a feasible solution do:
        \begin{enumerate}
            \item Add a new vertex $a^{i,X,d,F}$ to $V(H)$
            \item Let $a^{j, X\cap B_j, d', F\cap E(B_j)}$ be the vertex of $V(H')$ guaranteed by property $2)$, such that $$(H_{a^{j, X\cap B_j, d', F\cap E(B_j)}},g,c)$$ is a dynamic programming formulation for $C[j,X\cap B_j, d', F\cap E(B_j)]$, where $$d'_u := d_u - \begin{cases} 1, &\text{ if $uv \in F$}\\
                                0, &\text{otherwise.}
            \end{cases}$$ for all $u \in B_j.$
            \item Add arc $e=(a^{i,X,d,F}, \{a^{j, X\cap B_j, d', F\cap E(B_j)}\})$ to $E(H)$.
            \item Define action of $g$ on $e$ to be 
            $$g(e) = \begin{pmatrix}
                F\backslash E(B_j) \\ 
                (X\backslash B_j) \cup (\{v\}\cap V(F))
            \end{pmatrix} \in \{0,1\}^{E(G_i), V(G_i)}$$
            \item Define action of $c$ on $e$ to be $$c(e) = w(F\backslash E(B_j)) - x(X\backslash B_j) - x(\{v\}\cap V(F)).$$
        \end{enumerate}
    \end{enumerate}
    Property $1)$ is immediate from induction and the number of arcs added by the construction.

    To verify property $2)$ consider $C[i',X,d,F]$ satisfying the hypothesis of property $2)$. If $i' \neq i$ property $2)$ holds immediately by induction. So suppose $i' = i$. We claim that $(H_{a^{i,X,d,F}}, g,c)$ is a dynamic programming formulation for $C[i,X,d,F]$. 
    
    Let $\begin{pmatrix}M& S\end{pmatrix}^T$ be a feasible solution to $C[i,X,d,F]$. For any $\begin{pmatrix}M & S\end{pmatrix}^T$ feasible for $C[i,X,d,F]$ there exists $
    \begin{pmatrix}M' & S'\end{pmatrix}^T$ feasible for $C[j, X \cap B_j, d', F\cap E(B_j)]$ (where $d'$ is as defined in Step $2b)$ of the construction) such that 
    $$M = M' \dot\cup \left(F\cap E(B_j)\right) \quad\text{and}\quad S = S' \dot\cup \left(X\backslash B_j) \cup (\{v\}) \cap V(F))\right).$$
    Moreover, for any $\begin{pmatrix}M' & S'\end{pmatrix}^T$ feasible for $C[j, X\cap B_j, d', F\cap E(B_j)]$,
    $$ \begin{pmatrix}
        M'\dot\cup \left(F\cap E(B_j)\right) \\
        S' \dot\cup \left(X\backslash B_j\right)\cup (\{v\}\cap V(F))
    \end{pmatrix}$$
    is feasible for $C[i,X,d,F]$. Hence by Step $2d)$ and Step $2e)$ of the construction, and the inductive hypothesis, $g$ and $c$ behave as desired for the dynamic program $(H_{a^{i,X,d,F}},g,c)$ to be a dynamic programming formulation of $C[i,X,d,F]$.

    \paragraph{Case: Forget.} If $i$ is a Forget node then $i$ has one child $j$, and $B_i\dot\cup \{v\} = B_j$ for some vertex $v \in V$. Let $(H', g', c')$ be the dynamic programming formulation for $C[j]$ guaranteed by the inductive hypothesis. We construct a new dynamic programming formulation $(H,g,c)$ via the following procedure:
    \begin{enumerate}
        \item Initialize $H$ to $H'$.
        \item For each $e' \in E(H')$, define the action of $g$ on $e'$ to be $g(e') = g'(e')$.
        \item For each $e' \in E(H')$, define the action of $c$ on $e'$ to be $c(e') = c'(e')$.
        \item For every $X\subseteq B_i, d\in Z^{B_i}, F\subseteq E(B_i)$ such that $C[i,X,d,F]$ has a feasible solution do:
        \begin{enumerate}
            \item Add a new vertex $a^{i,X,d,F}$ to $V(H)$.
            \item For each $J \subseteq \delta(v) \cap E(B_j)$, $d' \in \Z^{B_j}$ such that $d'_u = d_u$ for all $u \in B_i$, and for each $Y \in \{X, X \cup\{v\}\}$, if $C[j,Y,d',F\cup J]$ has a feasible solution do:
            \begin{enumerate}
                \item Let $a^{j, Y, d', F\cup J}$ be the vertex of $V(H')$ guaranteed by property 2), such that $(H_{a^{j, Y, d', F\cup J}}, g,c)$ is a dynamic programming formulation for $C[j, Y, d', F\cup J]$.
                \item Add arc $e = (a^{i,X,d,F}, \{a^{j,Y,d', F\cup J}\})$ to $E(H)$.
                \item Define action of $g$ on $e$ to be $g(e) = 0$.
                \item Define action of $c$ on $e$ to be $c(e) = -w(J) + x(Y\cap \{v\})$.
            \end{enumerate}
        \end{enumerate}
    \end{enumerate}
    Property $1)$ is immediate from induction and the number of arcs added by the construction.

    To verify property $2)$ consider $C[i',X,d,F]$ satisfying the hypothesis of property $2)$. If $i'\neq i$ property $2)$ holds immediately by induction. So suppose $i' = i$. We claim that $(H_{a^{i,X,d,F}}, g,c)$ is a dynamic programming formulation for $C[i,X,d,F]$.

    Let $\begin{pmatrix}
        M &  S
    \end{pmatrix}^T$ be feasible for $C[i,X,d,F]$. Then there exists $$\begin{pmatrix}
        M' \\ S'
    \end{pmatrix} \in \{0,1\}^{E(G_j),V(G_j)}$$ such that $M' = M \dot\cup J$ for some $J \in \delta(v) \cap E(B_j)$ and $S' = S \cup Y$ for some $Y \in \{X, X \cup \{v\}\}$. Further there exists $d' \in \Z^{B_j}$ such that $|M'\cap \delta(u)| = d'_u$ for all $u \in B_j$. Hence $(M',S')$ is feasible for $C[j,Y,d',F\cup J]$. Moreover for any $\begin{pmatrix}
        M' & S'
    \end{pmatrix}^T$ feasible for $C[j,Y,d',F\cup J]$,
    $$\begin{pmatrix}
        M'\backslash J \\
        (S\cup X)\backslash Y
    \end{pmatrix}$$
    is feasible for $C[i,X,d,F]$. Thus by induction and Step $2b)$ of the construction, $g$ and $c$ behave as desired for $(H_{a^{i,X,d,F}}, g,c)$ to be a dynamic programming formulation of $C[i,X,d,F]$.

    \paragraph{Case: Join.} If $i$ is a Join node then $i$ has two children: $j_1$ and $j_2$, and $B_{j_1}=B_{j_2}=B_i$. Let $(H^1,g^1,c^1)$ be the dynamic programming formulation for $C[j_1]$ guaranteed by the induction hypothesis, and let $(H^2,g^2,c^2)$ be the similarly guaranteed dynamic programming formulation for $C[j_2]$ We construct a new dynamic programming formulation $(H,g,c)$ via the following procedure:
    \begin{enumerate}
        \item Initialize $H = H^1\dot\cup H^2$.
        \item For each $e' \in E(H^1)$, define the action of $g$ on $e'$ to be $g(e') = g^1(e')$. Similarly for each $e' \in E(H^2)$, define the action of $g$ on $e'$ to be $g(e') = g^2(e')$.
        \item For each $e' \in E(H^1)$, define the action of $c$ on $e'$ to be $c(e') = c^1(e')$. Similarly for each $e' \in E(H^2)$, define the action of $c$ on $e'$ to be $c(e') = c^2(e')$.
        \item For every $X \subseteq B_i$, $d \in \Z^{B_i}$, $F \subseteq E(B_i)$ such that $C[i,X,d,F]$ has a feasible do:
        \begin{enumerate}
            \item Add a new vertex $a^{i,X,d,F}$ to $V(H)$.
            \item For every $d^1 \in \Z^{B_{j_1}}$ such that $C[j_1, X,d^1, F]$ has a feasible solution, and for every $d^2 \in \Z^{B_{j_2}}$ such that $C[j_2, X, d^2, F]$ has a feasible solution, if $d^1+d^2-|F\cap \delta(u)| = d_u$ for all $u \in B_i$ then do:
            \begin{enumerate}
                \item Let $a^{j_1, X, d^1, F}$ be the vertex of $V(H^1)$ guaranteed by property $2)$, such that $(H_{a^{j_1, X,d^1,F}}, g,c)$ is a dynamic programming formulation for $C[j_1,X,d^1,F]$.
                \item Let $a^{j_2, X, d^2, F}$ be the vertex of $V(H^2)$ guaranteed by property $2)$, such that $(H_{a^{j_2, X,d^2,F}}, g,c)$ is a dynamic programming formulation for $C[j_2,X,d^2,F]$.
                \item Add arc $e=(a^{i,X,d,F}, \{a^{j_1, X, d^1, F}, a^{j_2, X, d^2, F}\})$ to $E(H)$.
                \item Define action of $g$ on $e$ to be $g(e) = 0$.
                \item Define action of $c$ on $e$ to be $$c(e) = -w(F) + x(\{u \in B_i : d^1_u > 0 \text{ and } d^2_u > 0\}).$$
            \end{enumerate}
        \end{enumerate}
    \end{enumerate}
    Property $1)$ is immediate from induction and the number of arcs added by the construction.

    To verify property $2)$ consider $C[i',X,d,F]$ satisfying the hypothesis of property $2)$. If $i'\neq i$ property $2)$ holds immediately by induction. So suppose $i'=i$. We claim that $(H_{a^{i,X,d,F}},g,c)$ is a dynamic programming formulation for $C[i,X,d,F]$.
    
    Let $\begin{pmatrix}M & S\end{pmatrix}^T$ be feasible for $C[i,X,d,F]$. Then there exists $\begin{pmatrix}M_1 & S_1\end{pmatrix}^T$ feasible for $C[j_1, X, d^1, F]$ and $\begin{pmatrix}M_2 & S_2\end{pmatrix}^T$ feasible for $C[j_2,X,d^2,F]$ such that $$d^1+d^2-|F\cap \delta(u)| = d_u$$ for all $u \in B_i$, satisfying that $$M = M_1 \cup M_2 \quad\text{and}\quad S = S_1 \cup S_2.$$
    Moreover for any $\begin{pmatrix}M_1 & S_1\end{pmatrix}^T$ feasible for $C[j_1,X,d^1,F]$ and for any $\begin{pmatrix}M_2 & S_2\end{pmatrix}^T$ feasible for $C[j_2,X,d^2,F]$ such that $d^1+d^2-|F\cap \delta(u)| = d_u$ for all $u \in B_i$, $$\begin{pmatrix}M_1\cup M_2\\S_1 \cup S_2\end{pmatrix}$$ is feasible for $C[i,X,d,F]$. Hence $g(\cP(H_{a^{i,X,d,F}}))$ is equal to the feasible region $C[i,X,d,F]$ by Steps $2b)ii-2b)iv$ of the construction. Furthermore 
    \begin{align*}
        w(M) - x(S) &= w(M_1) + w(M_2) - w(M_1\cap M_2) - x(S_1) - x(S_2) + x(S_1 \cap S_2)\\
        &=w(M_1) -x(S_1) + w(M_2) - x(S_2) - w(F)\\ &+ x(\{u \in B_i : d^1_u > 0 \text{ and } d^2_u > 0\})
    \end{align*}
    with the second equality following since there are no edges from a vertex in $V(G_{j_1})\backslash B_i$ to a vertex in $V(G_{j_2})\backslash B_i$ by the properties of tree decompositions. Hence $c$ behaves as desired for $(H_{a^{i,X,d,F]}}, g,c)$ to be a dynamic programming formulation of $C[i,X,d,F]$.

    It is not hard to see that the constructions in each case preserve integrality.
$\blacksquare$\end{proof}

By the preceding discussion, Lemma~\ref{lemma:min-excess-b-matching-specific-dp}, and Theorem~\ref{th:dp-nucleolus} we have shown Theorem~\ref{th:b-matching-nucleolus}.

\section{Conclusion and Future Work}
\paragraph{}
We have given a formalization of dynamic programming, and shown that in this formal model adding congruency constraints only increases the complexity by a polynomial factor of the prime modulus. From this, we showed that whenever the minimum excess coalition problem of a cooperative game can be solved via dynamic programming, its nucleolus can be computed in time polynomial in the size of the dynamic program. Using this result we gave an algorithm for computing the nucleolus of $b$-matching games on graphs of bounded treewidth.

In~\cite{Kaibel2010a} they show that a generalization of the dynamic programming model in~\cite{martin1990polyhedral} called Branched Polyhedral Systems also has an integral extended formulation. It is natural to wonder how our framework could extend to Branched Polyhedral Systems and if that would enable to computation of the nucleolus for any interesting classes of cooperative games.

\bibliography{library}
\bibliographystyle{splncs04}

%
%
%

\end{document}